\documentclass[
  aps,
  pra,
  twocolumn,
  superscriptaddress,
  nofootinbib,
  longbibliography
]{revtex4-2}

\usepackage{amsmath,amssymb,amsfonts,amsthm}
\usepackage{mathtools}
\usepackage{braket}
\usepackage{graphicx}
\usepackage{xcolor}
\usepackage{hyperref}
\usepackage{booktabs}
\usepackage{tikz}
\usepackage{tabularx}
\usepackage[ruled,vlined]{algorithm2e}
\DontPrintSemicolon

\newtheorem{theorem}{Theorem}
\newtheorem{lemma}{Lemma}
\newtheorem{proposition}{Proposition}

\newtheorem{definition}{Definition}

\newcommand{\Tr}{\operatorname{Tr}}
\newcommand{\supp}{\operatorname{supp}}
\newcommand{\Odd}{\operatorname{Odd}}

\newcommand{\Ftwo}{\mathbb{F}_2}
\newcommand{\ketbra}[1]{\ket{#1}\!\bra{#1}}

\newcommand{\SA}{\mathcal S_A}
\newcommand{\calS}{\mathcal S}

\begin{document}

\title{A five-qubit 1-resistant graph state and stabilizer marginal certificates}

\author{Zicheng Han}
\email{hzcqj2020@mail.ustc.edu.cn}
\affiliation{School of Mathematical Sciences, University of Science and Technology of China, Hefei 230026, China}

\author{Wanchen Zhang}
\email{wanchenz@mail.ustc.edu.cn}
\affiliation{Hefei National Laboratory, University of Science and Technology of China, Hefei 230088, China}

\author{Xiande Zhang}
\email{drzhangx@ustc.edu.cn}
\thanks{Corresponding author}
\affiliation{School of Mathematical Sciences, University of Science and Technology of China, Hefei 230026, China}
\affiliation{Hefei National Laboratory, University of Science and Technology of China, Hefei 230088, China}


\begin{abstract}
We study particle-loss resistant entanglement within the framework of
stabilizer and graph states. A pure state is \(m\)-resistant if it remains
entangled after the loss of any \(m\) particles and becomes fully separable
after the loss of any \(m+1\) particles. The smallest previously unresolved
qubit case was the existence of a five-qubit \(1\)-resistant pure state, which is resolved here by the five-cycle graph state \(\ket{C_5}\). A
stabilizer-subgroup method is also developed for verifying \(m\)-resistance
in graph states, using local stabilizers to certify full separability and
exact negative partial transpose~(NPT) witnesses to certify entanglement.


Applying this to all graph states associated with non-isomorphic graphs on
five, six, and seven vertices, we obtain a graph state classification up to local Clifford equivalence, which also classifies stabilizer states 
 up to local Clifford equivalence. 
Thus, the five-qubit \(1\)-resistant
stabilizer states are exactly the local Clifford class of \(C_5\).
Six-qubit \(2\)-resistant stabilizer states exist in three distinct
local Clifford classes, whereas no seven-qubit stabilizer state is \(m\)-resistant
for any nonzero admissible \(m\). Finally, we prove that the cycle
graph states \(\ket{C_N}\) with \(N\ge 7\) are not \(m\)-resistant for any
\(0\le m\le N-2\).
\end{abstract}

\maketitle

\section{Introduction}

Quantum entanglement is a fundamental feature that distinguishes quantum systems from classical ones and serves as a key resource for quantum information processing, quantum communication, and quantum computation \cite{Horodecki2009, GuhneToth2009, Steane1998}. It characterizes nonclassical correlations among subsystems, whose states cannot in general be described independently even when they are spatially separated. In multipartite systems, an important question is how such correlations behave when some particles are lost, measured, or become inaccessible \cite{VidalTarrach1999, Dur2001, Briegel2001,Rajagopal2002, Sugita2007, Brunner2012, Quinta2018,Neven2018, Quinta2019,Zhang2025}. Mathematically, particle loss is described by taking the partial trace, which maps the original multipartite state to reduced states of the remaining subsystems~\cite{NielsenChuang2010}. Depending on the initial entanglement structure, the resulting reduced states may either become separable or retain certain forms of entanglement \cite{VidalTarrach1999, Dur2001, Briegel2001, Sugita2007,Rajagopal2002, Sugita2007, Brunner2012, Balasubramanian2017, Quinta2018,Neven2018, Quinta2019,Luo2021, Zhang2025}. Understanding whether and how entanglement survives or transforms under subsystem loss is therefore a central problem in the study of multipartite entanglement \cite{Briegel2001,Brunner2012, Quinta2018,Neven2018, Quinta2019, Zhang2025}.

To characterize the robustness of multipartite entanglement under particle loss, Quinta et al.~\cite{Quinta2019} introduced the notion of \(m\)-resistant states. An \(N\)-particle entangled pure state is called \(m\)-resistant if it remains entangled after the loss of any \(m\) particles, while it becomes fully separable after the loss of any \(m+1\) particles. This notion captures a sharp transition between entanglement robustness and fragility. Quinta et al.~\cite{Quinta2018,Quinta2019} further connected this problem with the theory of cut-resistant topological links, where cutting link components plays a role analogous to losing particles in a multipartite quantum system. Based on this analogy, they proposed a feasible construction scheme for general \(m\)-resistant mixed states. For pure states, however, the construction problem is substantially more challenging: they conjectured that, for any \(N\)-particle system, there should exist an \(m\)-resistant pure state for some local dimension \(d\), and proved this conjecture for \(m\leq (N-2)/2\). More recently, Zhang et al.~\cite{Zhang2025} introduced the stronger notion of \(m\)-resistance, where the surviving entanglement after any \(m\)-particle loss is required to be genuinely multipartite~\cite{GuhneSeevinck2010}, and developed two general constructions of \(m\)-resistant pure states based on Dicke-state superpositions~\cite{Dicke1954} and classical error-correcting codes. In the qubit setting, the smallest unresolved case was the existence of a five-qubit \(1\)-resistant pure state. Existing general constructions cover \(m=0,2,3\) for \(N=5\), but not \(m=1\). Moreover, symmetric-state searches failed to find such a state, suggesting that it might not exist~\cite{Quinta2019, Zhang2025}.

In this work, we show that the five-cycle graph state resolves this case~\cite{Quinta2019,Zhang2025}. More
precisely, we prove that \(\ket{C_5}\) is a \(1\)-resistant five-qubit
pure state: all of its three-qubit marginals are fully separable, whereas all
of its four-qubit marginals are entangled.

We also develop a stabilizer-subgroup framework for verifying
\(m\)-resistance of graph states, building on the stabilizer and graph-state
formalism~\cite{Gottesman1997,Hein2004,Hein2006Review}. Full separability is certified by
local stabilizer structure, while entanglement is certified by exact NPT
witnesses, based on the positive partial transpose criterion
\cite{Peres1996,Horodecki1996}. This gives a certificate-based method, rather
than a floating-point numerical search.

The remainder of the paper applies this method to graph states of small size.
We first classify the five-vertex case, showing that the new \(1\)-resistant
example is unique up to local Clifford equivalence, or equivalently up to
local complementation of graphs~\cite{VanDenNest2004}. We then analyze the
six-vertex and seven-vertex cases, where the possible resistance parameters
exhibit a mixture of existence and nonexistence results. The \(m=0\) cases are
also included and reduce, within these graph-state classes, to the GHZ
local-complementation orbits~\cite{GHZ1989,Quinta2019}. Finally, we prove that the positive cycle
examples \(C_5\) and \(C_6\) do not extend to larger cycles: for every
\(N\ge 7\), \(\ket{C_N}\) is not \(m\)-resistant for
\(0\le m\le N-2\).

Although the enumeration is carried out in the graph-state representation, the
resulting classifications also have a stabilizer-state interpretation. Indeed,
every stabilizer state is local Clifford equivalent to some graph state~\cite{VanDenNest2004,Hein2006Review}, and
\(m\)-resistance is invariant under local unitary transformations. Therefore,
the graph-state classifications obtained in this paper lift directly to classifications
of stabilizer states up to local Clifford equivalence.
\section{Preliminaries}

\subsection{Particle loss and \(m\)-resistance}
For a positive integer \(N\), let $[N]:=\{1,2,\ldots,N\}.$
For \(J\subseteq [N]\), denote its complement by
$\bar J=[N]\setminus J.$
Given an \(N\)-partite pure state \(\ket{\psi}\), we write
\[
\rho_{\bar J}
=
\Tr_J \ketbra{\psi}
\]
for the reduced state obtained after tracing out the particles in \(J\).
We use the terms reduced state and marginal interchangeably.

A mixed state \(\rho\) on
$\mathcal H_1\otimes\cdots\otimes \mathcal H_N$
is called fully separable if it can be written as
\[
\rho
=
\sum_\alpha p_\alpha\,
\rho^{(\alpha)}_1\otimes
\rho^{(\alpha)}_2\otimes\cdots\otimes
\rho^{(\alpha)}_N,
\]
where \(p_\alpha\ge 0\), \(\sum_\alpha p_\alpha=1\), and each
\(\rho^{(\alpha)}_i\) is a one-particle density operator~\cite{guhne2011entanglement}. A state that is not
fully separable is called entangled in this paper.

\begin{definition}[\(m\)-resistance~\cite{Quinta2019}]
Let \(\ket{\psi}\) be an \(N\)-partite pure state. The state \(\ket{\psi}\) is
called \(m\)-resistant if:
\begin{enumerate}
    \item for every \(J\subseteq [N]\) with \(|J|=m\), the reduced state
    \(\rho_{\bar J}\) is entangled;
    \item for every \(J\subseteq [N]\) with \(|J|=m+1\), the reduced state
    \(\rho_{\bar J}\) is fully separable.
\end{enumerate}
\end{definition}

Throughout this paper, \(m\)-resistance refers to the above notion unless
otherwise specified.

For \(N=5\) and \(m=1\), the definition requires that every four-qubit
marginal is entangled, whereas every three-qubit marginal is fully separable.
This is the criterion used below to verify the \(1\)-resistance of the
five-cycle graph state.

\subsection{Graph states and stabilizers}

We first recall the stabilizer-state formalism
\cite{Gottesman1997,Hein2004,Hein2006Review}. Let \(X,Y,Z\) denote the
Pauli matrices, and let \(I\) be the identity operator. An \(N\)-qubit Pauli
operator is a tensor product
\[
P=P_1\otimes P_2\otimes\cdots\otimes P_N,
\qquad
P_i\in\{I,X,Y,Z\},
\]
up to an overall phase. The support of \(P\), denoted by
\(\operatorname{supp}(P)\), is the set of qubits on which \(P\) acts
nontrivially:
\[
\operatorname{supp}(P)=\{i\in [N]: P_i\neq I\}.
\]

A stabilizer state is a pure state uniquely specified as the common \(+1\)
eigenstate of an abelian subgroup of the \(N\)-qubit Pauli group. More
precisely, if \(g_1,\ldots,g_N\) are \(N\) independent commuting Pauli
operators, then the state \(\ket{\psi}\) satisfying
\[
g_j\ket{\psi}=\ket{\psi},
\qquad j=1,\ldots,N,
\]
is called a stabilizer state. The group
\[
\mathcal S=\langle g_1,\ldots,g_N\rangle
\]
is its stabilizer group. It contains \(2^N\) Pauli operators, and the
rank-one projector onto \(\ket{\psi}\) has the stabilizer expansion
\[
\ketbra{\psi}
=
2^{-N}\sum_{g\in\mathcal S}g.
\]

Graph states form an important class of stabilizer states~\cite{Hein2004,Hein2006Review}. Let \(G=(V,E)\) be a simple graph with \(|V|=N\).
When convenient we identify \(V\) with \([N]\). 
The graph state
\(\ket{G}\) is obtained by preparing each qubit in
\[
\ket{+}=\frac{\ket{0}+\ket{1}}{\sqrt 2}
\]
and applying a controlled-\(Z\) gate along each edge:
\[
\ket{G}
=
\prod_{\{u,v\}\in E} CZ_{uv}\ket{+}^{\otimes N}.
\]
Here \(CZ_{uv}\) denotes the controlled-\(Z\) gate acting on qubits \(u\) and
\(v\).

Equivalently, \(\ket{G}\) is the unique common \(+1\) eigenstate of the
operators~\cite{Hein2004,Hein2006Review}
\[
K_v
=
X_v\prod_{u\in N(v)} Z_u,
\qquad v\in V,
\]
where \(N(v)\) denotes the neighborhood of \(v\) in \(G\). Here \(X_v\) and
\(Z_u\) denote Pauli operators acting on the corresponding qubits, with the
identity acting on all other qubits.

The operators \(K_v\) commute with each other and are independent. Hence they
generate the stabilizer group of the graph state,
\[
\mathcal S(G)=\langle K_v:v\in V\rangle .
\]
Thus, by the general stabilizer expansion above,
\[
\ketbra{G}
=
2^{-N}\sum_{g\in\mathcal S(G)}g.
\]

\section{Stabilizer marginals and certificates}
We now record two sufficient certificates for reduced stabilizer states: one
for full separability and one for entanglement.

Let \(\ket{\psi}\) be an
\(N\)-qubit stabilizer state with stabilizer group \(\calS\). For
\(A\subseteq[N]\), define the local stabilizer subgroup
\[
\SA
=
\{g\in\calS:\supp(g)\subseteq A\}.
\]
For \(g\in\SA\), we denote by \(g|_A\) the restriction of the Pauli string
\(g\) to the tensor factors in \(A\). Equivalently, since \(g\) acts trivially
on \(\bar A\), we have
\[
g=g|_A\otimes I_{\bar A}.
\]

Using the stabilizer expansion~\cite{Gottesman1997}, taking the partial trace over
\(\bar A\) removes every Pauli string that is nontrivial on at least one
qubit in \(\bar A\), since
$\Tr X=\Tr Y=\Tr Z=0.$
Hence only the elements of \(\SA\) survive, and the reduced state on \(A\) is
\begin{equation}
\rho_A
=
\Tr_{\bar A}\ketbra{\psi}
=
2^{-|A|}
\sum_{g\in\SA}g|_A .
\label{eq:stabilizer-marginal}
\end{equation}
Thus, the study of stabilizer marginals is reduced to the study of the local
stabilizer subgroup \(\SA\).

\subsection{A separability certificate}

We first give a sufficient condition for full separability. This is the main separability certificate used in our graph-state search. This elementary certificate uses only the local Pauli structure of the stabilizer subgroup and follows the generator-based viewpoint of Refs.~\cite{Fattal2004, AudenaertPlenio2005}.

Throughout this paper, the dimension of a stabilizer subgroup means its
dimension as a vector space over \(\Ftwo\), or equivalently the number of
independent stabilizer generators.

\begin{lemma}[Sufficient condition for full separability]
\label{lem:local-diagonal-separable}
Let \(\ket{\psi}\) be an \(N\)-qubit stabilizer state, and let
\(\mathcal S_A\) be the local stabilizer subgroup supported on
\(A\). Suppose that, for every qubit \(j\in A\), all non-identity
single-qubit Pauli factors appearing at site \(j\) among the elements of
\(\mathcal S_A\) are of one fixed Pauli type. 
Then the reduced state \(\rho_A\) is fully separable.

In particular, if
\[
\dim_{\Ftwo}\SA\le 1,
\]
then \(\rho_A\) is fully separable.
\end{lemma}

\begin{proof}
By assumption, for each qubit \(j\in A\), there exists a single-qubit basis
in which all non-identity Pauli factors appearing at site \(j\) among the
elements of \(\SA\) are diagonal. Taking the tensor product of these
single-qubit bases gives a product basis on \(A\) in which every element of
\(\SA\) is diagonal. Hence, by Eq.~\eqref{eq:stabilizer-marginal}, the
reduced state \(\rho_A\) is diagonal in this product basis. Therefore
\(\rho_A\) is a classical mixture of product basis projectors, and is fully
separable.

It remains to explain the special case
\(\dim_{\mathbb F_2}\SA\le1\). If \(\SA=\{I\}\), then
$\rho_A=\frac{I}{2^{|A|}}$,
which is fully separable. If \(\dim_{\mathbb F_2}\SA=1\), then
$\SA=\{I,P\}$,
where
\[
P|_A=P_1\otimes\cdots\otimes P_r,
\qquad r=|A|,
\]
and each \(P_j\) is a single-qubit Pauli operator on the \(j\)th qubit of
\(A\). Then
\[
\rho_A
=
2^{-r}
\left(I+P_1\otimes\cdots\otimes P_r\right).
\]
Let
\[
\Pi_j^\pm=\frac{I\pm P_j}{2}.
\]
The following explicit decomposition holds:
\[
2^{-r}
\left(I+P_1\otimes\cdots\otimes P_r\right)
=
\frac{1}{2^{r-1}}
\sum_{\epsilon_1\cdots\epsilon_r=+1}
\Pi_1^{\epsilon_1}\otimes\cdots\otimes \Pi_r^{\epsilon_r}.
\]
The right-hand side is a convex combination of product projectors. Hence
\(\rho_A\) is fully separable.
\end{proof}

Lemma~\ref{lem:local-diagonal-separable} is a sufficient condition, not a
complete separability criterion. When the lemma applies, the corresponding
marginal is fully separable. When it does not apply, no
conclusion about separability or entanglement is drawn from this lemma alone.

\subsection{An NPT entanglement certificate}

We next describe the entanglement certificate used in our search. It is based
on the positive partial transpose criterion~\cite{Peres1996,Horodecki1996,TothGuhne2005Stabilizer}.

Let \(g_1,\ldots,g_s\) be independent generators of \(\SA\),  where
$s=\dim_{\Ftwo}\SA$.
For \(u=(u_1,\ldots,u_s)\in\Ftwo^s\), define
\[
g(u)=g_1^{u_1}\cdots g_s^{u_s}.
\]
Then the elements of \(\SA\) are precisely the operators \(g(u)\), with
\(u\in\Ftwo^s\). Thus Eq.~\eqref{eq:stabilizer-marginal} gives
\[
\rho_A
=
2^{-|A|}
\sum_{u\in\Ftwo^s}g(u).
\]

Let \(R\subseteq A\) be one side of a bipartition
$R\,|\,(A\setminus R)$.
We denote by \(\Gamma_R\) the partial-transpose map acting on the qubits in
\(R\), with respect to the computational basis. Thus, for an operator \(M\)
on the subsystem \(A\), \(M^{\Gamma_R}\) is obtained by transposing the tensor
factors in \(R\) and leaving those in \(A\setminus R\) unchanged.

Under \(\Gamma_R\), a Pauli string changes sign exactly when it contains an
odd number of \(Y\)'s on the transposed qubits, because
$X^\top=X, Z^\top=Z, Y^\top=-Y$.
For \(u\in\Ftwo^s\), define
\[
Y_R(u)
:=
\{j\in R:\ g(u)\text{ acts as }Y\text{ on qubit }j\}.
\]
Then
$g(u)^{\Gamma_R}
=
(-1)^{|Y_R(u)|}g(u)$.
For convenience, write
\[
\tau_R(u):=(-1)^{|Y_R(u)|}.
\]
Therefore
\begin{equation}
    \rho_A^{\Gamma_R}=
2^{-|A|}
\sum_{u\in\Ftwo^s}
\tau_R(u)g(u).
\label{eq:partial-transpose}
\end{equation}

Since \(g_1,\ldots,g_s\) commute, they can be simultaneously diagonalized.
A joint eigenspace is labelled by a character \(x\in\Ftwo^s\), on which
$g(u)\mapsto (-1)^{x\cdot u}$.
Consequently, all eigenvalues of \(\rho_A^{\Gamma_R}\) are
\begin{equation}
\lambda_x(R)
=
2^{-|A|}
\sum_{u\in\Ftwo^s}
\tau_R(u)(-1)^{x\cdot u},
\qquad x\in\Ftwo^s .
\label{eq:walsh-eigenvalues}
\end{equation}

Thus, the partial-transpose spectrum is obtained as the Walsh transform of
the sign function \(\tau_R\). This is a direct character diagonalization of a
stabilizer-diagonal operator. Similar sign-change calculations under partial
transposition are standard for graph-diagonal and GHZ-diagonal states
\cite{Kay2010GraphDiagonal,Kay2011}. In the present work, we use the
corresponding form for local stabilizer marginals.

If there exist \(R\subseteq A\) and \(x\in\Ftwo^s\) such that
\[
\lambda_x(R)<0,
\]
then \(\rho_A^{\Gamma_R}\) is not positive semidefinite. Hence \(\rho_A\) is
NPT across the bipartition \(R\ |\ (A\setminus R)\), and therefore \(\rho_A\) is
entangled. This is a sufficient certificate for entanglement. Failure to find
such a negative eigenvalue does not imply separability.

\section{Graph-state formulation}
\label{sec:graph-state-formulation}
We now specialize the above two certificates to graph states~\cite{Hein2004,Hein2006Review}. Let
\(G=(V,E)\) be a graph on \(N\) vertices. For each \(v\in V\), the graph-state
stabilizer generator is
\[
K_v
=
X_v\prod_{u\in N(v)}Z_u,
\]
where \(N(v)\) is the neighborhood of \(v\). For a subset \(U\subseteq V\),
define
\[
K_U=\prod_{v\in U}K_v.
\]
The graph-state stabilizer group is 
$\calS(G)=\{K_U:U\subseteq V\}$.

For \(U\subseteq V\), define the odd neighborhood
\[
\Odd_G(U)
:=
\{w\in V: |N(w)\cap U|\equiv 1 \pmod 2\}.
\]
The support of \(K_U\) is
\begin{equation}
\supp(K_U)=U\cup\Odd_G(U).
\label{eq:support-KU}
\end{equation}
Indeed, a vertex \(w\) receives an \(X_w\) factor precisely when \(w\in U\),
and it receives a \(Z_w\) factor once for each neighbor of \(w\) that lies in
\(U\). Hence a \(Z_w\) factor remains exactly when \(w\in\Odd_G(U)\).

For a subsystem \(A\subseteq V\), the stabilizer elements that
survive the partial trace over \(V\setminus A\) are precisely those satisfying
$\supp(K_U)\subseteq A$.
By Eq.~\eqref{eq:support-KU},  we define the binary index space
\begin{equation}
L_A
:=
\{U\subseteq V: U\cup \Odd_G(U)\subseteq A\}. \label{eq:LA}
\end{equation}
Hence, the local stabilizer subgroup of the graph-state marginal
on \(A\) is
\[
\mathcal S_A(G)
=
\{K_U:U\in L_A\}.
\]
Consequently, the marginal on \(A\) is
\begin{equation}
\rho_A
=
2^{-|A|}
\sum_{U\in L_A}K_U|_A.
\label{eq:graph-marginal}
\end{equation}

Eq.~\eqref{eq:graph-marginal} gives a concrete finite procedure for checking
\(m\)-resistance of graph states. 
More concretely, we check two requirements through the local
stabilizer subgroup \(\mathcal S_A(G)\):

\begin{enumerate}
    \item For each subset \(A\) with \(|A|=N-m-1\), compute
    \(\mathcal S_A(G)\). If \(\mathcal S_A(G)\) satisfies the local
    stabilizer separability criterion of
    Lemma~\ref{lem:local-diagonal-separable}, then \(\rho_A\) is certified
    fully separable. In particular,
    \[
   \dim_{\mathbb F_2}\mathcal S_A(G)= \dim_{\mathbb F_2}L_A\le 1
    \]
    is sufficient. Here, we identify subsets in  $L_A$ as binary vectors, and treat $L_A$ as a linear subspace of $\mathbb{F}_2^{N}$, since the mapping $U \mapsto \Odd_G(U)$ is linear over $\mathbb{F}_2$.

    \item For each subset \(A\) with \(|A|=N-m\), compute
    \(\mathcal S_A(G)\) and search over bipartitions
    \(R\,|\,(A\setminus R)\). If Eq.~\eqref{eq:walsh-eigenvalues} gives
    \[
    \lambda_x(R)<0,  \text{ for some } x\in\Ftwo^s
    \]
    then \(\rho_A\) is  entangled by the NPT criterion.
\end{enumerate}

Thus the graph-state verification is certificate-based: full separability is
proved by a local stabilizer certificate for \(\mathcal S_A(G)\), while
entanglement is proved by an exact NPT certificate computed from the same local
stabilizer subgroup. In the small graph-state classifications below, every
reported solution or failure is obtained from these two certificates.

\section{The five-cycle graph state is \(1\)-resistant}

Let $C_5$ be the five-cycle graph with vertices $\mathbb Z_5$ and edges
$(i,i+1), i\in\mathbb Z_5$.
The corresponding graph state is
\[
\ket{C_5}
=
\prod_{i=0}^4 CZ_{i,i+1}\ket{+}^{\otimes 5}.
\]
Its stabilizer generators are
\[
K_i=Z_{i-1}X_iZ_{i+1},
\qquad i\in\mathbb Z_5.
\]

\begin{theorem}
\label{thm:C5}
The five-cycle graph state $\ket{C_5}$ is a five-qubit  $1$-resistant pure state.
\end{theorem}

\begin{proof}
We first show that every three-qubit marginal is fully separable.
By the dihedral symmetry of \(C_5\), there are two orbits of three-vertex
subsets, represented by
\[
A_1=\{0,1,2\},
\qquad
A_2=\{0,1,3\}.
\]
For these two representatives, the local stabilizer subgroups are
\[
\mathcal S_{A_1}(C_5)
=
\langle K_1\rangle
=
\langle Z_0X_1Z_2\rangle,
\]
and
\[
\mathcal S_{A_2}(C_5)
=
\langle K_0K_1K_3\rangle
=
\langle Y_0Y_1X_3\rangle.
\]
Thus, in both cases,
\[
\dim_{\mathbb F_2}\mathcal S_{A_i}(C_5)=1
\qquad (i=1,2).
\]
By Lemma~\ref{lem:local-diagonal-separable}, the corresponding marginals are
fully separable. Since every three-qubit marginal is equivalent to one of
these two cases by a graph automorphism of \(C_5\), all three-qubit marginals
are fully separable.

It remains to show that every four-qubit marginal is entangled.
By cyclic symmetry, it suffices to consider
\[
\rho_{0123}
=
\Tr_4\ketbra{C_5}.
\]
The stabilizer subgroup supported on $\{0,1,2,3\}$ is generated by
\[
\begin{aligned}
a &= Z_0X_1Z_2I_3,\\
b &= I_0Z_1X_2Z_3,\\
c &= X_0Z_1Z_2X_3.
\end{aligned}
\]

Therefore,
\[
\rho_{0123}
=
\frac1{16}(I+a)(I+b)(I+c).
\]
Applying the partial transpose on qubit \(0\), namely \(\Gamma_{\{0\}}\), and
using Eq.~\eqref{eq:partial-transpose}, we obtain
\[
\rho_{0123}^{\Gamma_{\{0\}}}
=
\frac1{16}
(I+a+b+c+ab-ac+bc-abc).
\]
Using Eq.~\eqref{eq:walsh-eigenvalues} with
$R=\{0\},$ and $
x=(1,0,1)\in\mathbb F_2^3,$
we obtain the eigenvalue
\[
\lambda_{(1,0,1)}(\{0\})
=
\frac1{16}
(1-1+1-1-1-1-1-1)
=
-\frac14.
\]
Thus \(\rho_{0123}^{\Gamma_{\{0\}}}\) is not positive semidefinite. Therefore
\(\rho_{0123}\) is NPT and hence entangled.

\end{proof}

\section{Small graph-state classification}

\label{sec:small-graph-classification}

We now apply the graph-state verification procedure of the previous section to
all graph states on five, six, and seven vertices. The results in this section
are finite exhaustive classifications within the class of graph states. They
should therefore be distinguished from statements about arbitrary pure states:
nonexistence within graph states does not exclude the possibility of
non-stabilizer-resistant pure states.

Before describing the enumeration, we record a simple invariance property that
justifies grouping graph-state solutions into orbits under local complementation.

\begin{proposition}[Local-unitary invariance of \(m\)-resistance]
\label{prop:lu-invariance}
The \(m\)-resistant property is invariant under local unitary operations. In particular,
it is invariant under local Clifford operations.
\end{proposition}

\begin{proof}
Let
$U=U_1\otimes\cdots\otimes U_N$
be a local unitary and let
$\ket{\phi}=U\ket{\psi}.$
For every \(J\subseteq[N]\),
the reduced state of \(\ket{\phi}\) on
\(\bar J\) is
\[
\rho_{\bar J}^{\phi}
=
U_{\bar J}
\rho_{\bar J}^{\psi}
U_{\bar J}^{\dagger}.
\]
Thus each marginal of \(\ket{\phi}\) is locally unitarily equivalent to the
corresponding marginal of \(\ket{\psi}\).

Local unitaries preserve full separability, and since they are reversible,
they also preserve non-full-separability. Hence the two defining conditions of
\(m\)-resistance are unchanged under local unitaries. Local Clifford
operations are local unitaries, so \(m\)-resistance is also invariant under
local Clifford operations.
\end{proof}

For graph states, local Clifford equivalence is represented graph-theoretically
by local complementation together with graph isomorphism
\cite{VanDenNest2004,VanDenNest2004Algorithm}. Therefore, by
Proposition~\ref{prop:lu-invariance}, if one graph in a
local-complementation orbit gives an \(m\)-resistant graph state, then every
graph in the same orbit gives an \(m\)-resistant graph state. This is why the
classification below is organized by local-complementation orbits.

\subsection{Enumeration protocol}
\label{subsec:enumeration-protocol}

All enumerations were performed over non-isomorphic simple \(N\)-vertex graphs. For
\(N=5,6,7\), the numbers of unlabeled graphs checked were
\[
34,\qquad 156,\qquad 1044,
\]
respectively. Each graph was stored as a canonical isomorphism representative.
For each graph \(G\) and each admissible value of \(m\), we applied the
certificate test described in Sec.~\ref {sec:graph-state-formulation} to all
relevant marginals.

More explicitly, for every subsystem \(A\), we computed the corresponding local
stabilizer subgroup \(\mathcal S_A(G)\) from the index space \(L_A\) defined in
Eq.~\eqref{eq:LA}. Full separability was certified from the structure of
\(\mathcal S_A(G)\) using Lemma~\ref{lem:local-diagonal-separable}. Entanglement was certified by searching over partial transposes and applying the exact NPT criterion of Eq.~\eqref{eq:walsh-eigenvalues}.

A graph was accepted as \(m\)-resistant only if all required marginal
certificates were found. For rejected candidates, the computation records an
explicit failure certificate: either a marginal required to be fully separable
is certified NPT, or a marginal required to be entangled is certified fully
separable.

Local-complementation orbits were generated by repeatedly applying local
complementation at each vertex and reducing the resulting graphs to canonical
isomorphism representatives. For graph states, local Clifford equivalence is characterized by sequences of local complementations and graph isomorphisms~\cite{VanDenNest2004, VanDenNest2004Algorithm}.

The implementation details of the enumeration protocol and the results are provided in Appendix~\ref{appendixA}.

We first  mention the \(0\)-resistant case from our
classification results. It is known that the GHZ state is a standard example of a
\(0\)-resistant pure state~\cite{Quinta2019}. In the
graph-state setting, the \(N\)-qubit GHZ state is local Clifford equivalent to
the star graph state on \(K_{1, N-1}\). In fact, our enumeration shows that, when $N=5,6,7$, this known
example is in fact the only \(0\)-resistant class: every \(0\)-resistant
 graph state  lies in the local-complementation orbit of
\(K_{1,N-1}\). 

When  $N=5,6,7$, we present our results separately for $m=1,\ldots,N-2$.



\subsection{Enumeration of five-qubit graph state}
For \(N=5\), we consider 
$m=1,2,3.$
When  \(m=1\), the \(1\)-resistance
requires all four-qubit marginals to be entangled and all three-qubit
marginals to be fully separable. By Theorem~\ref{thm:C5}, the five-cycle graph
state \(|C_5\rangle\) satisfies exactly these conditions.

The exhaustive graph-state classification strengthens this example as follows.

\begin{proposition}
\label{prop:five-vertex-classification}
The \(1\)-resistant five-vertex graph states are exactly the
local-complementation orbit of \(C_5\). No five-vertex graph state is
\(m\)-resistant for \(m=2,3\).
\end{proposition}

Thus, within graph states, the five-qubit \(1\)-resistant solution is unique
up to local Clifford equivalence. Up to graph isomorphism, the
local-complementation orbit of \(C_5\) contains three representatives:
\[
C_5,
\]
the graph obtained from \(C_5\) by adding one chord, and
\[
K_5\setminus P_4.
\]

For \(m=2\), one would require all three-qubit marginals to be entangled and
all two-qubit marginals to be fully separable. The enumeration shows that
every five-vertex graph state has at least one three-qubit marginal certified
to be fully separable, and hence cannot be \(2\)-resistant. For \(m=3\), all
two-qubit marginals would have to be entangled, while the one-qubit marginals
are automatically fully separable. However, every five-vertex graph state has
at least one two-qubit marginal certified to be fully separable. Therefore
\(3\)-resistance is also impossible within five-vertex graph states.

\subsection{Enumeration of six-qubit graph state}
\label{subsec:six-qubit-graph-states}

For \(N=6\), the admissible
resistance parameters are
$m=1,2,3,4.$
These resistance parameters exhibit a different pattern. The cases
\(m=1,3,4\) admit no graph-state solution, whereas the intermediate case
\(m=2\) admits graph-state solutions.

\begin{proposition}
\label{prop:six-nonexistence}
There are no six-vertex graph states that are \(m\)-resistant for
\(m=1,3,4\).
\end{proposition}

The remaining case \(m=2\) is positive. Here resistance requires all
four-qubit marginals to be entangled and all three-qubit marginals to be fully
separable. This behavior is already known for the six-qubit AME state:
its three-qubit marginals are maximally mixed, while its four-qubit marginals
remain entangled~\cite{Helwig2012QSS, Helwig2013Graph, Zhang2025}. Our
enumeration shows that, within six-vertex graph states, the AME orbit is only
one of three possible local-complementation orbits.

\begin{proposition}
\label{prop:six-existence}
There are \(23\) non-isomorphic six-vertex \(2\)-resistant graph states.
They split into three local-complementation orbits, represented by
\(G_{\mathrm I}\), \(C_6\), and \(G_{\mathrm{AME}}\) in
Fig.~\ref{fig:six-representatives}.
\end{proposition}

Thus, the six-qubit AME graph-state orbit recovers a known \(2\)-resistant
construction~\cite{Zhang2025}, while the full enumeration exhibits two additional
local-complementation orbits. By Proposition~\ref{prop:lu-invariance}, every
graph in these three orbits is also \(2\)-resistant.

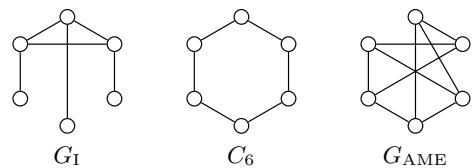
\begin{figure}[htbp]
\centering
\begin{tikzpicture}[
    scale=0.72,
    vtx/.style={circle, draw, fill=white, inner sep=0pt, minimum size=5.5pt},
    every path/.style={line width=0.45pt}
]

\begin{scope}[shift={(0,0)}]
\node[vtx] (a0) at (90:1.0) {};
\node[vtx] (a1) at (30:1.0) {};
\node[vtx] (a2) at (-30:1.0) {};
\node[vtx] (a3) at (-90:1.0) {};
\node[vtx] (a4) at (-150:1.0) {};
\node[vtx] (a5) at (150:1.0) {};

\draw (a0)--(a5);
\draw (a1)--(a5);
\draw (a0)--(a1);
\draw (a3)--(a0);
\draw (a4)--(a5);
\draw (a1)--(a2);

\node[draw=none, fill=none, rectangle] at (0,-1.55) {\small \(G_{\mathrm I}\)};
\end{scope}

\begin{scope}[shift={(3.2,0)}]
\node[vtx] (b0) at (90:1.0) {};
\node[vtx] (b1) at (30:1.0) {};
\node[vtx] (b2) at (-30:1.0) {};
\node[vtx] (b3) at (-90:1.0) {};
\node[vtx] (b4) at (-150:1.0) {};
\node[vtx] (b5) at (150:1.0) {};

\draw (b0)--(b1);
\draw (b1)--(b2);
\draw (b2)--(b3);
\draw (b3)--(b4);
\draw (b4)--(b5);
\draw (b5)--(b0);

\node[draw=none, fill=none, rectangle] at (0,-1.55) {\small \(C_6\)};
\end{scope}

\begin{scope}[shift={(6.4,0)}]
\node[vtx] (c0) at (90:1.0) {};
\node[vtx] (c1) at (30:1.0) {};
\node[vtx] (c2) at (-30:1.0) {};
\node[vtx] (c3) at (-90:1.0) {};
\node[vtx] (c4) at (-150:1.0) {};
\node[vtx] (c5) at (150:1.0) {};

\draw (c0)--(c1);
\draw (c0)--(c2);
\draw (c0)--(c3);
\draw (c1)--(c4);
\draw (c1)--(c5);
\draw (c2)--(c3);
\draw (c2)--(c5);
\draw (c3)--(c4);
\draw (c4)--(c5);

\node[draw=none, fill=none, rectangle] at (0,-1.55) {\small \(G_{\mathrm{AME}}\)};
\end{scope}

\end{tikzpicture}
\caption{
Unlabeled representatives of the three local-complementation orbits of
six-vertex \(2\)-resistant graph states. 
The right representative \(G_{\mathrm{AME}}\) is a
graph-state realization of the known six-qubit AME \(2\)-resistant state.
}
\label{fig:six-representatives}
\end{figure}

\subsection{Enumeration of seven-qubit graph state}

For \(N=7\), the admissible resistance parameters are
$m=1,2,3,4,5.$
For every nonzero resistance parameter $m$,
the exhaustive certificate-based enumeration finds no seven-vertex graph-state
solution. For each graph and each value of \(m\), the computation produces an
explicit obstruction: either a marginal required to be fully separable is
certified NPT, or a marginal required to be entangled is certified fully
separable.

\begin{proposition}
\label{prop:seven-vertex-classification}
No seven-vertex graph state is
\(m\)-resistant for \(m=1,2,3,4,5\).
\end{proposition}

\subsection{Summary and stabilizer-state interpretation}

The small graph-state classification is summarized in
Table~\ref{tab:small-classification}. The table lists the number of
non-isomorphic graphs checked, the number of graph-state solutions up to graph
isomorphism, and the number of local-complementation orbits.

\begin{table*}[htbp]
\centering
\caption{Small graph-state classification for \(m\)-resistance.
Here ``Checked'' denotes the number of non-isomorphic \(N\)-vertex graphs
examined.}
\label{tab:small-classification}
\begin{tabular}{c c c c l}
\toprule
\((N,m)\)
& Checked
& Solutions
& LC orbits
& Classification \\
\midrule
\((5,0)\) & \(34\) & \(2\) & \(1\)
& GHZ-type LC orbit of \(K_{1,4}\) \\
\((5,1)\) & \(34\) & \(3\) & \(1\)
& LC orbit of \(C_5\) \\
\((5,m),\ m=2,3\) & \(34\) & \(0\) & \(0\)
& No graph-state solution \\
\((6,0)\) & \(156\) & \(2\) & \(1\)
& GHZ-type LC orbit of \(K_{1,5}\) \\
\((6,1)\) & \(156\) & \(0\) & \(0\)
& No graph-state solution \\
\((6,2)\) & \(156\) & \(23\) & \(3\)
& LC orbits represented by \(G_{\mathrm I}\), \(C_6\), and \(G_{\rm AME}\) \\
\((6,m),\ m=3,4\) & \(156\) & \(0\) & \(0\)
& No graph-state solution \\
\((7,0)\) & \(1044\) & \(2\) & \(1\)
& GHZ-type LC orbit of \(K_{1,6}\) \\
\((7,m),\ m=1,\ldots,5\) & \(1044\) & \(0\) & \(0\)
& No graph-state solution \\
\bottomrule
\end{tabular}
\end{table*}

These graph-state results also have a direct stabilizer-state interpretation.
Indeed, every \(N\)-qubit stabilizer state is local Clifford equivalent to a
graph state~\cite{VanDenNest2004}. Since \(m\)-resistance is invariant under
local unitary operations by Proposition~\ref{prop:lu-invariance}, the
classification in Table~\ref{tab:small-classification} can equivalently be
read as a classification of stabilizer states up to local Clifford
equivalence.

\section{Cycle graph states of larger size}

The cycle graph states \(C_5\) and \(C_6\) play a special role in the small
graph-state classification above. The five-cycle graph state gives a
five-qubit  \(1\)-resistant state, while the six-cycle graph state
gives a six-qubit  \(2\)-resistant state. However, our exhaustive
search shows that the seven-cycle graph state does not yield an 
\(m\)-resistant state for any admissible nonzero value of \(m\). This naturally raises the question of whether larger cycle graph states \(C_N\) can be 
\(m\)-resistant for some \(m\). The following result shows that this is not
the case once the cycle is sufficiently large. Thus the positive cycle examples \(C_5\) and \(C_6\) do not extend to an
infinite family.

\begin{proposition}[Non-resistance of large cycle graph states]
\label{prop:cycle-obstructions}
For every \(N\ge 7\), the cycle graph state \(\ket{C_N}\) is not
\(m\)-resistant for any
$0\le m\le N-2.$
\end{proposition}

\begin{proof} 
It suffices to consider \(N\ge 8\), since the case \(N=7\) follows from the exhaustive classification above.
Label the vertices of \(C_N\) by \(\mathbb Z_N\), and write the graph-state
stabilizer generators as
\[
K_i=Z_{i-1}X_iZ_{i+1},
\qquad i\in\mathbb Z_N .
\]
Suppose that \(\ket{C_N}\) were \(m\)-resistant, and set
\[
k=N-m-1.
\]
Then \(1\le k\le N-1\), and \(m\)-resistance requires every \(k\)-qubit
marginal to be fully separable and every \((k+1)\)-qubit marginal to be
entangled.

We first rule out the case \(k\ge4\). Consider the four consecutive vertices
$B=\{0,1,2,3\}.$
The local stabilizer subgroup supported on \(B\) is generated by
\[
a=Z_0X_1Z_2I_3,
\qquad
b=I_0Z_1X_2Z_3 .
\]
Thus
\[
\rho_B
=
\frac{1}{16}(I+a)(I+b).
\]
Applying the partial transpose on qubit \(1\), namely \(\Gamma_{\{1\}}\), and
using Eq.~\eqref{eq:partial-transpose}, only the term \(ab\)
changes sign. 
Hence
\[
\rho_B^{\Gamma_{\{1\}}}
=
\frac{1}{16}(I+a+b-ab).
\]
Equivalently, using the Walsh eigenvalue formula
Eq.~\eqref{eq:walsh-eigenvalues} with
$R=\{1\},$ and $
x=(1,1)\in\mathbb F_2^2,$
we obtain
\[
\lambda_{(1,1)}(\{1\})
=
\frac{1}{16}(1-1-1-1)
=
-\frac18 .
\]
Therefore \(\rho_B^{\Gamma_{\{1\}}}\) is not positive semidefinite, so
\(\rho_B\) is NPT and hence entangled.

If \(k\ge4\), choose a \(k\)-element subset \(A\) such that
\[
B\subseteq A.
\]
Since \(\ket{C_N}\) is assumed to be \(m\)-resistant, the marginal
\(\rho_A\) should be fully separable. But every marginal of a fully separable
state is fully separable, so
\[
\rho_B=\Tr_{A\setminus B}\rho_A
\]
would also be fully separable. This contradicts the entanglement of
\(\rho_B\). Hence no \(m\)-resistance is possible when \(k\ge4\).

It remains to consider \(k=3,2,1\). In each case, we exhibit a fully separable
\((k+1)\)-qubit marginal, contradicting the requirement that all
\((k+1)\)-qubit marginals be entangled.

For \(k=3\), take
$B=\{0,2,4,6\}.$
If \(N>8\), no nontrivial stabilizer element is supported on \(B\), and hence
\[
\rho_B=\frac{I}{16}.
\]
If \(N=8\), the only nontrivial stabilizer element supported on \(B\) is
generated by \(K_0K_2K_4K_6\). Thus, in both cases,
\[
\dim_{\mathbb F_2}\mathcal S_B(C_N)\le1.
\]
By Lemma~\ref{lem:local-diagonal-separable}, \(\rho_B\) is fully separable.
This contradicts the requirement that every four-qubit marginal be entangled.

For \(k=2\), take the three consecutive vertices
$B=\{0,1,2\}.$
The only nontrivial stabilizer element supported on \(B\) is
\[
K_1=Z_0X_1Z_2.
\]
Hence
\[
\dim_{\mathbb F_2}\mathcal S_B(C_N)=1,
\]
and Lemma~\ref{lem:local-diagonal-separable} implies that \(\rho_B\) is fully
separable. This contradicts the requirement that every three-qubit marginal
be entangled.

For \(k=1\), take two non-adjacent vertices, for example
$B=\{0,2\}.$
For \(N\ge8\), no nontrivial stabilizer element is supported on \(B\). Hence
\[
\rho_B=\frac{I}{4},
\]
which is fully separable. This contradicts the requirement that every
two-qubit marginal be entangled.

Thus, for every possible value \(1\le k\le N-1\), one of the two defining
requirements of \(m\)-resistance is violated. Therefore \(\ket{C_N}\) is not
\(m\)-resistant for any
$0\le m\le N-2.$
\end{proof}

\section{Discussion}

We have shown that the five-cycle graph state is a five-qubit
\(1\)-resistant pure state. This resolves the five-qubit case by a simple
stabilizer-state construction. More generally, we developed a
stabilizer-subgroup framework for verifying \(m\)-resistance of graph states
and applied it to small graph states. Within graph states, the five-qubit
\(1\)-resistant solutions form exactly the local-complementation orbit of
\(C_5\); six-qubit \(2\)-resistant graph states exist and split into three
local-complementation orbits; and no seven-vertex graph state is
\(m\)-resistant for \(m=1,\ldots,5\). We also showed that the positive cycle
examples \(C_5\) and \(C_6\) do not extend to an infinite family, that is,
\(\ket{C_N}\) is not \(m\)-resistant for any \(0\le m\le N-2\) when
\(N\ge 7\).

Several open problems remain. First, do \(C_5\) and \(C_6\) give strongly \(m\)-resistant graph states for \(m=1\) and \(m=2\), respectively? Here, ``strong'' means genuine multipartite entanglement rather than mere entanglement~\cite{Zhang2025,GuhneSeevinck2010}. Second, can non-stabilizer pure states realize resistance parameters that are ruled out within graph states, especially in the six- and seven-qubit cases?

Finally, we mention that the nonexistence results in this paper should be understood as no-go results
within the graph-state and stabilizer-state settings, rather than as a
nonexistence result for all pure states. For example, for \(4\)- and
\(5\)-resistance, seven-qubit pure states are known to exist~\cite{Zhang2025}. This shows that graph states do not capture all possible
particle-loss entanglement structures of general pure states.

\begin{acknowledgments}
The research of Xiande Zhang is supported by the National Key Research and Development Program of China 2023YFA1010200, the NSFC under Grants No. 12171452 and No.12231014, and the Quantum Science and Technology-National Science and Technology Major Project 2021ZD0302902. The research of Wanchen Zhang is supported by the Innovation Program for Quantum Science and Technology under grant no. 2021ZD0301602.
\end{acknowledgments}

\section*{DATA AVAILABILITY}

The source code and certificate data supporting the small graph-state classifications in this work are publicly available on Zenodo at DOI: \href{https://doi.org/10.5281/zenodo.20476315}{10.5281/zenodo.20476315}. The archive contains the Python enumeration script, human-readable classification outputs, and machine-readable JSON certificate files for all pairs \((N,m)\) with \(N=5,6,7\) considered in Table~\ref{tab:small-classification}. For each graph, the certificate files record the NetworkX graph-atlas identifier, edge set, resistance status, marginal certificates, and local-complementation orbit information when applicable. For every pair \((N,m)\), the enumeration returned either \textsc{RESISTANT} or \textsc{FAIL} for each graph; no \textsc{UNKNOWN} cases occurred.

\appendix
\section{Enumeration protocol and certificate data}
\label{appendixA}

\subsubsection{The implementation details}
We enumerate non-isomorphic \(N\)-vertex simple graphs for
\(N=5,6,7\) using NetworkX~\cite{NetworkX2008}. In what follows, we use the
corresponding NetworkX graph-atlas labels, such as \(\texttt{atlas\_14}\), to
identify the enumerated graph representatives. These labels are used only as
implementation identifiers; the graph itself is specified by its edge set, so
the classification is independent of the particular labeling convention used
by NetworkX.

The NetworkX enumeration only provides the list of graph representatives.
For each representative \(G\), the resistance property is then tested by the
certificate procedure developed in the main text. More precisely, for each
candidate graph we examine two families of marginals: the marginals of size
\(N-m-1\), which should be fully separable, and the marginals of size
\(N-m\), which should be entangled. The first family is checked using the
local stabilizer separability criterion, while the second family is checked
using the exact NPT certificate obtained from the Walsh-transform formula. If all
required marginals are certified, the graph is accepted; if one required
condition is contradicted by the opposite certificate, the graph is rejected;
otherwise the result is marked as inconclusive. The procedure is summarized
in Algorithm~\ref{alg:graph-search}.

No \textsc{UNKNOWN} cases occurred: for every candidate graph, either the \(m\)-resistance property was certified by the corresponding positive certificate, or it was explicitly rejected by a failure certificate.

\begin{algorithm}[htbp]
\caption{Certificate-based graph-state search}
\label{alg:graph-search}
\KwIn{Number of vertices \(N\), resistance parameter \(m\)}
\KwOut{Graph states certified to be  \(m\)-resistant}

Set
\[
s_{\rm sep}=N-m-1,\qquad s_{\rm ent}=N-m .
\]

\ForEach{non-isomorphic graph \(G\) on \(N\) vertices}{
    Set \(\mathrm{status}(G)\leftarrow \mathrm{RESISTANT}\)\;

    \ForEach{subset \(A\subseteq V(G)\) with \(|A|=s_{\rm sep}\)}{
        Compute the local stabilizer subgroup
        \[
        \mathcal S_A(G)=\{K_U:\supp(K_U)\subseteq A\}.
        \]
        \If{\(\rho_A\) is certified fully separable by the local stabilizer criterion Lemma~\ref{lem:local-diagonal-separable}}{
            continue\;
        }
        \If{\(\rho_A\) has an NPT certificate}{
            Set \(\mathrm{status}(G)\leftarrow \mathrm{FAIL}\)\;
            break\;
        }
        Set \(\mathrm{status}(G)\leftarrow \mathrm{UNKNOWN}\)\;
        break\;
    }

    \If{\(\mathrm{status}(G)\neq \mathrm{RESISTANT}\)}{
        continue\;
    }

    \ForEach{subset \(A\subseteq V(G)\) with \(|A|=s_{\rm ent}\)}{
        Compute the local stabilizer subgroup
        \[
        \mathcal S_A(G)=\{K_U:\supp(K_U)\subseteq A\}.
        \]
        \If{\(\rho_A\) is certified fully separable by the local stabilizer criterion}{
            Set \(\mathrm{status}(G)\leftarrow \mathrm{FAIL}\)\;
            break\;
        }
        \If{\(\rho_A\) has an NPT certificate}{
            continue\;
        }
        Set \(\mathrm{status}(G)\leftarrow \mathrm{UNKNOWN}\)\;
        break\;
    }

    \If{\(\mathrm{status}(G)=\mathrm{RESISTANT}\)}{
        Record \(G\) as an \(m\)-resistant graph state\;
    }
}
\end{algorithm}

\subsubsection{Existence cases and local-complementation orbits}

We list only the existence cases. For each pair \((N,m)\), we give the
NetworkX atlas representatives of the \(m\)-resistant graph states and their
decomposition into local-complementation orbits.

\paragraph{\((N,m)=(5,0)\).}
There are two non-isomorphic \(0\)-resistant graph states. They form one
local-complementation orbit:
\[
\mathcal O_{5,0}^{(1)}
=
\{\texttt{atlas\_29},\texttt{atlas\_52}\}.
\]

\paragraph{\((N,m)=(5,1)\).}
There are three non-isomorphic \(1\)-resistant graph states. They form one
local-complementation orbit:
\[
\mathcal O_{5,1}^{(1)}
=
\{\texttt{atlas\_38},\texttt{atlas\_43},\texttt{atlas\_47}\}.
\]

\paragraph{\((N,m)=(6,0)\).}
There are two non-isomorphic \(0\)-resistant graph states. They form one
local-complementation orbit:
\[
\mathcal O_{6,0}^{(1)}
=
\{\texttt{atlas\_77},\texttt{atlas\_208}\}.
\]

\paragraph{\((N,m)=(6,2)\).}
There are \(23\) non-isomorphic \(2\)-resistant graph states. They split into
three local-complementation orbits:
\[
\begin{aligned}
\mathcal O_{6,2}^{(1)}
=
\{&
\texttt{atlas\_105},\texttt{atlas\_127},\texttt{atlas\_128},\\
&
\texttt{atlas\_147},\texttt{atlas\_148},\texttt{atlas\_151},\\
&
\texttt{atlas\_152},\texttt{atlas\_163},\texttt{atlas\_164},\\
&
\texttt{atlas\_166},\texttt{atlas\_167},\texttt{atlas\_171},\\
&
\texttt{atlas\_180},\texttt{atlas\_184},\texttt{atlas\_188},\\
&
\texttt{atlas\_196}
\},
\end{aligned}
\]
\[
\begin{aligned}
\mathcal O_{6,2}^{(2)}
=
\{&
\texttt{atlas\_94},\texttt{atlas\_115},\texttt{atlas\_143}, \\
&
\texttt{atlas\_158},\texttt{atlas\_195}
\},
\end{aligned}
\]
and
\[
\mathcal O_{6,2}^{(3)}
=
\{\texttt{atlas\_174},\texttt{atlas\_187}\}.
\]

\paragraph{\((N,m)=(7,0)\).}
There are two non-isomorphic \(0\)-resistant graph states. They form one
local-complementation orbit:
\[
\mathcal O_{7,0}^{(1)}
=
\{\texttt{atlas\_270},\texttt{atlas\_1252}\}.
\]


\bibliographystyle{apsrev4-2}
\bibliography{references}

\end{document}